\documentclass{article}

\usepackage{url}

\usepackage{graphicx,amssymb,amsmath}
\usepackage[retainorgcmds]{IEEEtrantools}
\usepackage{lineno}
\usepackage{color}
\usepackage{wasysym}
\usepackage{microtype}

\newcommand{\famsym}{\mathcal{F}}
\newcommand{\pert}{\pi_{\epsilon}}
\newcommand{\Conv}{\operatorname{Conv}}

\newtheorem{theorem}{Theorem}
\newtheorem{lemma}[theorem]{Lemma}
\newtheorem{corollary}[theorem]{Corollary}
\newtheorem{obs}{Observation}
\newtheorem{defini}{Definition}

\newcommand{\LabelQuote}[2]{\vspace{0.5cm}%
     \parbox{10cm}{\em #1}\hspace*{0.5cm}(#2)\\[0.5cm]}

\def\QED{\ensuremath{{\square}}}

\def\markatright#1{\leavevmode\unskip\nobreak\quad\hspace*{\fill}{#1}}
\newenvironment{proof}
  {\begin{trivlist}\item[\hskip\labelsep{\emph{Proof.}}]}
  {\markatright{\QED}\end{trivlist}}

\newenvironment{namedproof}[1]
  {\begin{trivlist}\item[\hskip\labelsep{\emph{#1.}}]}
  {\markatright{\QED}\end{trivlist}}

\usepackage{thm-restate}

\usepackage[dvipsnames]{xcolor}
\usepackage[color=green!20]{todonotes}

\renewcommand{\phi}{\varphi}

\renewcommand{\epsilon}{\varepsilon}
\newcommand{\R}{\mathbb{R}}

\newcommand\blfootnote[1]{%
	\begingroup
	\renewcommand\thefootnote{}\footnote{#1}%
	\addtocounter{footnote}{-1}%
	\endgroup
}
\def\inst#1{$^{#1}$}

\usepackage{wrapfig}

\begin{document}

\title{No Selection Lemma for Empty Triangles}

\author{Ruy Fabila-Monroy\inst{1}
\and Carlos Hidalgo-Toscano\inst{1}
\and Daniel Perz\inst{2}  
\and Birgit~Vogtenhuber\inst{2}}

\date{}

\maketitle

\begin{center}
\inst{1}Departamento de Matem\'aticas, Cinvestav, Ciudad de M\'exico, M\'exico, \\
ruyfabila@math.cinvestav.edu.mx, cmhidalgo@math.cinvestav.mx 
\\\ \\
\inst{2}Institute of Software Technology, Graz University of Technology, Graz, Austria, \\
daperz@ist.tugraz.at, bvogt@ist.tugraz.at
\end{center}

\begin{abstract}
Let $S$ be a set of $n$ points in general position in the plane. 
The Second Selection Lemma states that for any family of $\Theta(n^3)$ triangles spanned by $S$, there exists a point of the plane that lies in a constant fraction of them.
For families of $\Theta(n^{3-\alpha})$ triangles, with $0\le \alpha \le 1$, there might not be a point in more than $\Theta(n^{3-2\alpha})$ of those triangles.
An  empty triangle of $S$ is a triangle spanned by $S$
not containing any point of $S$ in its interior. B\'ar\'any conjectured that there exist an edge
spanned by $S$ that is incident to a super constant number of empty triangles of $S$. The number of empty triangles
of $S$ might be $O(n^2)$; in such a case, on average, every edge spanned by $S$ is incident to  a constant number
of empty triangles. The conjecture of B\'ar\'any suggests that for the class of empty triangles the above upper bound
might not hold. In this paper we show that, somewhat surprisingly,
the above upper bound does in fact hold for empty triangles. Specifically, we show that for any integer $n$ and real number $0\leq \alpha \leq 1$ there exists a point set of size $n$ with $\Theta(n^{3-\alpha})$ empty triangles such that any point of the plane is only in $O(n^{3-2\alpha})$ empty triangles.

\blfootnote{
	\begin{minipage}[l]{0.22\textwidth} \vspace{-6pt}\hspace{-10pt}\includegraphics[trim=10cm 6cm 10cm 5cm,clip,scale=0.155]{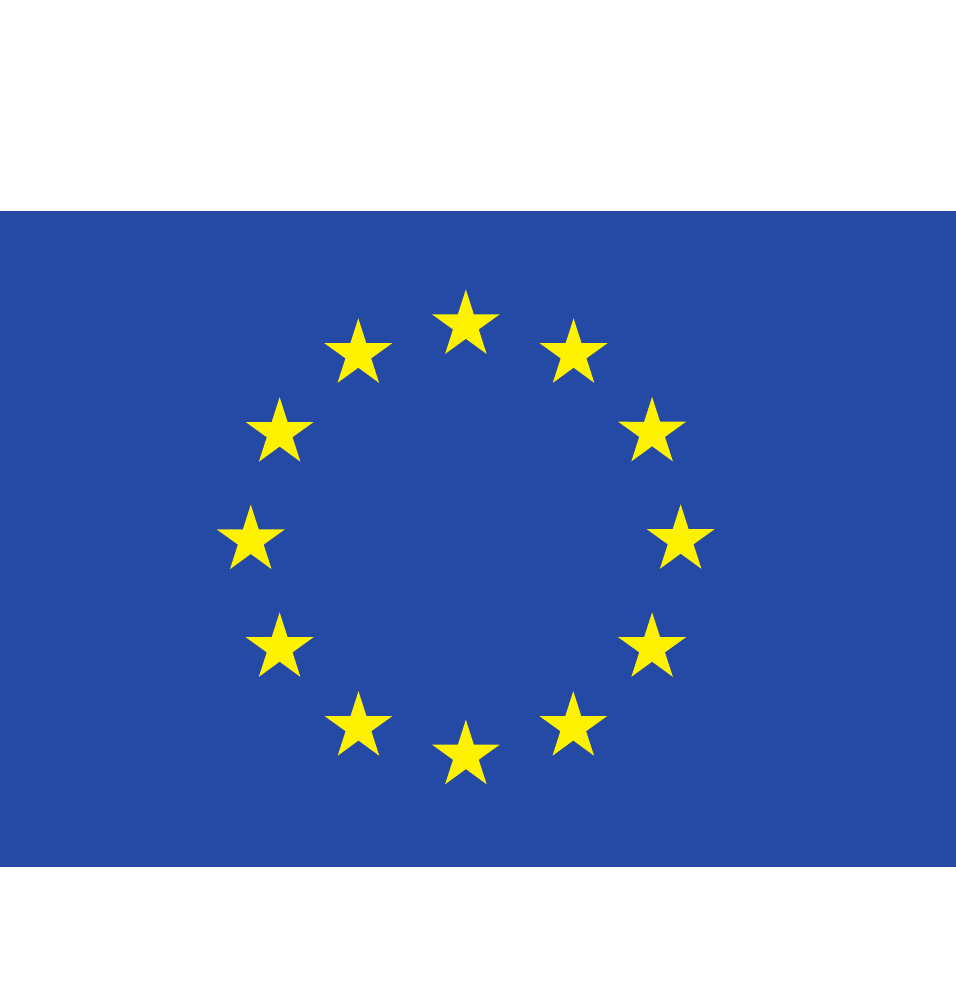} \end{minipage} \hspace{-1.55cm} \begin{minipage}[l][1cm]{0.855\textwidth} 
		\vspace{-1pt}
		This project has received funding from the European Union's Horizon 2020 research and innovation programme under the Marie Sk\l{}odowska-Curie grant agreement No 734922. \end{minipage}

\hspace{-23pt} 
\begin{minipage}[l]{\textwidth}
D.P.\, and B.V.\, were partially supported by the Austrian Science Fund within the 
	collaborative DACH project \emph{Arrangements and Drawings} as FWF project \mbox{I 3340-N35}.\\[-1ex]
\end{minipage}
}

\end{abstract}

\section{Introduction}
Let $S$ be a set of $n$ points in general position\footnote{A point set $S \subset \R^d$ is in general position if for every integer $1 < k \le d+1 $, no subset of $k$ points of $S$ is contained in a $(k\!-\!2)$-dimensional flat.}
 in the plane.
A \emph{triangle of $S$} is a triangle whose vertices are points of $S$.
We say that a point $p$ of the plane \emph{stabs} a triangle $\Delta$ if it lies in the interior 
of $\Delta$.
Boros and F\"uredi~\cite{boros1984number} showed that for any point set $S$ in general position in the plane, there exists a point in the plane which stabs a
constant fraction ($\frac{n^3}{27}+O(n^2)$) of the triangles of~$S$.
B\'ar\'any~\cite{BARANY-82} extended the result to $\R^d$; 
he showed that there exists a constant $c_d>0$, 
depending only on $d$, such that for any 
point set $S_d \subset \R^d$ in general position,
there exists a point in 
$\R^d$ which is in the interior of $c_d n^{d+1}$ $d$-dimensional simplices spanned by $S_d$. 
This result is known as 
\emph{First Selection Lemma}~\cite{MatousekBook-2002}.

Later, researchers considered the problem of the existence of a point in many triangles of a given family, $\famsym$, of triangles of $S$.
B\'ar\'any, F\"uredi and Lov\'asz~\cite{barany1990number} showed that for any point set $S$ in the plane in general position and any family $\famsym$ of $\Theta(n^3)$ triangles of $S$, there exists a point of the plane which stabs $\Theta(n^3)$ triangles from~$\famsym$.
This result, generalized to $\R^d$ 
by Alon et al.~\cite{alon1992point}, is 
now also known as 
\emph{Second Selection Lemma}~\cite{MatousekBook-2002}.

Both results for the plane require families of $\Theta(n^3)$ triangles of $S$.
It is  
natural to ask about families of triangles of smaller cardinality. 
For this question, Aronov et al.~\cite{SecondSelectionAronov} showed that for every $0 \leq \alpha \leq 1$ and every family $\famsym$ of $\Theta(n^{3-\alpha})$ triangles of $S$, there exists a point of the plane which stabs $\Omega(n^{3-3\alpha}/\log^5 n)$ triangles of $\famsym$.
This lower bound was 
improved by Eppstein~\cite{EppsteinImprovedBound} to the maximum of $n^{3-\alpha}/(2n-5)$ and $\Omega(n^{3-3\alpha}/\log^2 n)$.
A mistake in one of the proofs was later found and fixed by Nivasch and Sharir~\cite{EppsteinRevisited}.
Furthermore, Eppstein~\cite{EppsteinImprovedBound} constructed $n$-point sets and families of $n^{3-\alpha}$ triangles in them such that every point of the plane is in at most $n^{3-\alpha}/(2n-5)$ triangles for $\alpha \geq 1$ and in at most $n^{3-2\alpha}$ triangles for $0\leq \alpha \leq 1$.
Hence, for the number of triangles of a family $\famsym$ that can be guaranteed to simultaneously contain some point of the plane,  there is a continuous transition from a linear fraction for $|\famsym|=O(n^2)$ to a constant fraction for $|\famsym|=\Theta(n^3)$.

A triangle of $S$ is said to be \emph{empty} if it does not contain any points of $S$ in its interior. Let $\tau(S)$ be the number of empty triangles of $S$.
It is easily shown that $\tau(S)$ is $\Omega(n^2)$; Katchalski and Meir~\cite{kat} showed that there exist $n$-point sets $S$ with $\tau(S)=\Theta(n^2)$. 
Note that for such point sets, an edge of $S$ 
is on average part of a constant number of empty triangles of $S$. 
However, B{\'a}r{\'a}ny conjectured  that there is always an edge of $S$ which is part of a super constant number of empty triangles of $S$;
see~\cite{high_degree_erdos,high_degree_barany}.
B{\'a}r{\'a}ny et al.~\cite{BaranyMR13} proved this conjecture for random $n$-point sets, showing that for such such sets,
$\Theta(n/\log n)$ empty triangles are expected to share an edge. Note that the expected total number of empty triangles in such point sets is $\Theta(n^2)$; see~\cite{Va95}.

B\'ar\'any's conjecture suggests that perhaps there is always a point of the plane stabbing many
empty triangles of $S$, for any set $S$ of $n$ points in general position.
Naturally, the mentioned lower bounds for the number of triangles stabbed by a point of the plane also apply for the family of all empty triangles of $S$. 
In contrast, the upper bound constructions of Eppstein do not apply, since they contain non-empty triangles or do not contain all empty triangles of their underlying point sets.  
In this paper, we show that the existence of a point in more triangles than these upper bounds 
for general families of triangles is not guaranteed; hence the title of our paper. Specifically, we prove the following.

\begin{restatable}{theorem}{emptySelectionLemma}
	\label{thm:LensSquaredHortonSets}
	For every integer $n$ and every $0\leq \alpha \leq 1$, 
	there exist sets $S$ of $n$~points with $\tau(S) = \Theta{(n^{3-\alpha})}$ empty triangles 
	where every point of the plane stabs 
	$O{(n^{3-2\alpha})}$ empty triangles of $S$.
\end{restatable}

To prove Theorem~\ref{thm:LensSquaredHortonSets} for $\alpha=1$, we utilize the so called \emph{Horton sets} and \emph{squared Horton sets}. 
Horton~\cite{horton_1983} constructed a family of arbitrary large sets without large empty convex polygons.
Valtr~\cite{Va92} generalized Horton's construction and named the resulting sets ``Horton sets''.
Squared Horton sets were defined by Valtr~\cite{Va92} (as set $A_k$ in Section~4). 
B\'ar\'any and Valtr~\cite{BV2004} showed that squared Horton sets of size $n$ span only $\Theta(n^2)$ empty triangles. 

\paragraph{Outline.}
The remainder of this paper is organized as follows: 
In Section~\ref{sec:HortonSets}, we  give the definition of Horton sets and show several properties of them that will be of use for later sections. 
Section~\ref{sec:squaredHortonSets} considers squared Horton sets and contains a proof of Theorem~\ref{thm:LensSquaredHortonSets} for the case $\alpha=1$ (Theorem~\ref{thm:horton_stabs}). 
And in Section~\ref{sec:diamondssquaredhortonset} we present a generalized construction based on squared Horton sets, which we analyze to prove Theorem~\ref{thm:LensSquaredHortonSets}. 

\section{Horton sets}
\label{sec:HortonSets}

Let $X$ be a set of $n$ points in the plane such that no two points have
the same $x$-coordinate. 
In the following, 
we consider the points of $X$ 
in increasing order of their $x$-coordinates.
We denote with $X_0$ 
the subset of $X$ that contains every second point of $X$ (w.r.t.\ the $x$-order of the points), starting with the leftmost point of $X$. 
Similarly, $X_1=X\backslash X_0$ is the subset of $X$ that contains every second point of $X$  
(and does not contain the leftmost point of $X$).
In other words, if the points of $X$ are labeled $\{p_0, p_1,\dots, p_{n-1}\}$ in increasing $x$-order, then $X_0=\{p_0, p_2,\dots,\}$ and $X_1=\{p_1, p_3,\dots \}$.
In general, for a binary string $b$, we denote as $X_b$ the subset of vertices of $X$ that is obtained by recursively applying the above splitting. 
For example, $X_{10}$ consists of every second point of $X_1$ 
and does not contain the leftmost point of $X_1$.

Now consider two point sets $X$ and $Y$ in the plane such that no two points of $X \cup Y$ have
the same $x$-coordinate. We say that $Y$ is \emph{high above}
$X$ if every line passing through two points of $Y$ is above every
point of $X$, and $X$ is \emph{deep below}
$Y$ if every line passing through two points of $X$ is below every
point of $Y$. 

Using the above notation, we can now define Horton sets.
\begin{defini}
Let $H$ be a set of $n$ points in general position in the plane, such that no two points of $H$
have the same $x$-coordinate. 
Then $H$ is a \textbf{Horton set} if
\begin{enumerate}
 \item $|H|\leq 2$; or 
 
 \item $|H|>2$, $H_{0}$ and $H_{1}$ are Horton sets, and $H_1$ is high above $H_0$ and $H_0$ is deep below $H_1$.
\end{enumerate}
\end{defini}

One classic way to obtain a Horton $H=S_k$ set with $2^k$ points is by starting with a set $S_1$ of two points on a horizontal line, and then iteratively duplicating it by adding a translated copy $S'_i$ of $S_i$, where $S'_i$ is translated to the right by exactly half the $x$-distance between the first two points of $S_i$ and the translation in $y$-direction is such that $S'_i$ lies high above $S_i$. In the resulting Horton set, all points are evenly spaced in $x$-direction.

The following observation states that Horton sets have nice subset properties. They are directly implied by their definition.

\begin{obs}\label{obs:hortonsubsets}
Let $H=\{p_0, \ldots, p_{n-1}\}$ be a Horton set with points labeled in increasing $x$-order. Then for any $0 \leq i \leq j \leq n-1$, the subset $\{p_i, p_{i+1}, \ldots, p_j\}$ of consecutive points in $x$-direction again forms a Horton set. Similarly, for any integer $k$ and $0 \leq i \leq kj \leq n-1$, the set $\{p_i, p_{i+k}, p_{i+2k}, \ldots, p_{i+jk}\}$ is again a Horton set.
\end{obs}

We remark that a linear transformation of a Horton set, like for example a rotation, might no longer
be a Horton set by the above definition. However, the combinatorial properties of these sets do not change.
Hence, for convenience, we still call them  
Horton sets.

To analyze properties of the empty triangles of Horton sets, we define visible edges in Horton sets.
Let $H=H_0 \cup H_1$ be a Horton set of at least 4 points.
We say that an edge $e=(p_i,p_j)$, with $p_i,p_j \in H_0$, is \emph{visible from above} 
if $p_k$ is below the line spanned by $p_i$ and $p_j$ for every $p_k \in H_0$ with $i < k < j$. 
Likewise,  an edge $e=(p_i,p_j)$, with $p_i,p_j \in H_1$, is \emph{visible from below} if
$p_k$ is above the line spanned by $p_i$ and $p_j$ for every $p_k \in H_1$ with $i < k < j$. 
An edge of $H$ is \emph{visible} if it is either visible from above or visible from below.

\begin{lemma} \label{lem:visible}
An edge $e$ spanned by two vertices of a Horton set $H$
	is visible from below (above) if and only if it is spanned by two consecutive vertices of $H_b$, 
	where~$b$ is a binary string consisting of a single $1$ followed by an arbitrary number of $0$s (a single $0$ followed by an arbitrary number of $1$s).
\end{lemma}
\begin{proof}
 For the first direction of the proof, let $(p_i,p_j)$ be an edge of $H$ that visible from below. 
 Then, by the definition of visibility, $p_i,p_j \in H_1$.
 Let $b'$ be the unique binary string such that $p_i,p_j \in H_{1b'}$
 and such that either $p_i \in H_{1b'0}$ and $p_j \in H_{1b'1}$, or $p_i \in H_{1b'0}$ and $p_j \in H_{1b'1}$. 
 Without loss of generality assume that $p_i \in H_{1b'0}$ and $p_j \in H_{1b'1}$. 

 Suppose first that $b'$ is of the form $b'=b_1 1 b_2$ for some binary strings $b_1$, $b_2$. 
 Let $p_k$ be the point of $H_{b_1}$ that lies between $p_i$ and $p_j$. Note that $p_k \in H_{b_10}$. 
 By the definition of Horton sets $p_k$ is below the line spanned by $p_i$ and $p_j$;
 this contradicts the assumption that $(p_i,p_j)$ is visible from below. 
 Thus, $b'$ is a binary string that only consists of $0$s. 
	
 Next, suppose that $p_i$ and $p_j$ are not consecutive vertices of $H_{1b'}$. 
 Then there exists a point $p_k \in H_{1b'0} \subset H_1$ that lies between $p_i$ and $p_j$. 
 Again, by the definition of Horton sets, $p_k$ is below the line spanned by $p_i$ and $p_j$, 
 which contradicts the assumption that $(p_i,p_j)$ is visible from below.
 Hence, for $b=1b'$, $p_i$ and $p_j$ are consecutive vertices in $H_{b}$.	
 The reasoning for an edge $(p_i,p_j)$ that is visible from above is analogous, which completes the first direction of the proof.

 For the other direction, let $p_i,p_j$ be two consecutive points in $H_b$ for some binary string $b$ consisting of a single $1$ followed by an arbitrary number of $0$s.
 We proceed by induction on the length of $b$. If $b$ is empty then there is no point between $p_i$ and $p_j$ in $H_1$.
 Suppose that $b$ has length at least one. Let $b:=b'0$. There is exactly one point $p_k$ between $p_i$ and $p_j$
 in $H_{b'}$. Thus, $p_i$ and $p_k$ are consecutive points in $H_{b'}$. Likewise, $p_k$ and $p_j$ are consecutive points in $H_{b'}$.
 Let $\ell_1$ and $\ell_2$ be the two lines spanned by $(p_i,p_k)$, and $(p_k,p_j)$, respectively.
 By induction there are no points in $H_{b'}$ between $p_i$ and $p_k$, and below $\ell_1$. Likewise,
 there are no points in $H_{b'}$  between $p_k$ and $p_j$, and below $\ell_2$.
 Since $p_k$ is in $H_{b'1}$, $p_k$ is above the line $\ell$ spanned by $p_i$ and $p_j$.
 Thus, there are no points in $H_b$ between $p_i$ and $p_j$, and below $\ell$; this implies that
 $(p_i,p_j)$ is an edge visible from below of $H$.

 An analogous argument shows that if $p_i$ and $p_j$ are two consecutive points in $H_b$ for some binary string $b$ consisting of a single $0$ followed by an arbitrary number of $1$s,
 then $(p_i,p_j)$ is an edge of $H$ visible from above, which completes the proof.
\end{proof}

Note that visible edges are of central relevance for empty triangles in Horton sets.
Consider an empty triangle $\Delta$ in $H$ with vertices in both $H_0$ and $H_1$ 
and let $(p_i,p_j)$ be the edge of $\Delta$ such that both $p_i$ and $p_j$ are in $H_0$ or in $H_1$.
Then $(p_i,p_j)$ is a visible edge of $H$:
Assume without loss of generality that $p_i, p_j \in H_0$ 
and suppose for a contradiction that $(p_i,p_j)$ is not a visible edge of $H$. 
Then there exist a $p_k \in H_0$ such that  $p_k$ is above the line spanned by $p_i$ and $p_j$, and $i < k < j$. 
Since $H_1$ is high above $H_0$ this implies that $p_k$ is in the interior of $\Delta$; thus $\Delta$ is not empty.

The following two statements on empty triangles in Horton sets are  useful for proving our main theorem.

\begin{restatable}{lemma}{hortoninterior}\label{lemma:hortoninterior}
 Let $H$ be a Horton set of $n$ points. 
 Then every point $q$ of the plane stabs $O(n \log n)$ empty triangles of $H$.
\end{restatable}

\begin{proof}
Assume that $q \in \Conv(H)$, as otherwise $q$ stabs no empty triangle of~$S$.  
Let $b$ be the binary string such $q \in \Conv(H_b)$, 
but $q \notin \Conv(H_{b0})$  and $q \notin \Conv(H_{b1})$.
Let $\Delta$ be an empty triangle of $H$ stabbed by $q$. Note that either
two vertices of $\Delta$ lie in $H_{b0}$ and one vertex in $H_{b1}$, or 
two vertices of $\Delta$ lie in $H_{b1}$ and one vertex in~$H_{b0}$. 
Let $(p_i,p_j)$ be the edge of $\Delta$ such that both $p_i,p_j$ are in
 $H_{b0}$, or both $p_i,p_j$ are in  $H_{b1}$. Let $p_k$ be the other vertex of~$\Delta$.
Recall that $(p_i,p_j)$ is a visible edge. Let $b'$ be a binary string as in Lemma~\ref{lem:visible}
such that $p_i$ and $p_j$ are two consecutive vertices in~$H_{b'}$. Note that the only two consecutive
vertices of $H_{b'}$, that together with $p_k$ form an empty triangle containing $q$, are $p_i$ and $p_j$.
The number possible values for $b'$ is at most $2 \log_2 n$, and the number of possible choices
for $p_k$ is at most $n$. Therefore, the number of empty triangles stabbed by~$q$ is $O(n \log n)$.
\end{proof}

\begin{restatable}{lemma}{hortonincident}\label{lemma:hortonincident}
Let $H$ be a Horton set of $n$ points.
Then every point of $H$ is incident to $O(n \log n)$ empty triangles of $H$.
\end{restatable}

\begin{proof}
Let $q \in H$. Let $\Delta$ be an empty triangle of $H$ containing $q$ as a vertex. Let $(p_i,p_j)$ be the edge of $\Delta$
that is a visible edge of $H$. Let $p_k$ be the vertex of $\Delta$ distinct from $p_i$
and $p_j$. Let $b$ be the binary string for $(p_i,p_j)$ as in Lemma~\ref{lem:visible}, such that
$p_i$ and $p_j$ are two consecutive points of $H_b$. Suppose that $q$ is equal
to one of $p_i$ and $p_j$. For a fixed $b$ there are at most two possible choices for the other vertex
of $(p_i,p_j)$, and at most $n/2$ possible choices for $p_k$.  
Suppose that $q$ is equal to $p_k$. Then for a fixed $b$ there is exactly one choice for $p_i$
and $p_j$. Since the number of possible values for $b$ is $O(\log n)$, 
there are at most $O( \log n)$ empty triangles of $H$ containing $q$ as a vertex.
\end{proof}

\section{Squared Horton sets}\label{sec:squaredHortonSets}

For $n$ being a squared integer, we denote with $G$ an integer grid of size $\sqrt{n} \times \sqrt{n}$.
(Otherwise, $G$ is a subset of an integer grid of size $\lceil\sqrt{n}\rceil \times \lceil\sqrt{n}\rceil$,
from which some consecutive points of the topmost row and possibly the leftmost column are removed to have $n$ points remaining.)
An \emph{$\varepsilon$-perturbation} of $G$ is a perturbation of $G$ where every point $p$ of $G$ is mapped to a point at distance at most $\varepsilon$ to $p$.

\begin{defini}
A \emph{squared Horton set} $H$ 
of size $n$ is a specific $\varepsilon$-perturbation 
of $G$ 
such that the following three properties hold.
\begin{enumerate} 
	\item Any triple of non-collinear points in $G$ keeps its orientation in $H$.
	\item 
		The points on any non-vertical line spanned by points of $G$ are perturbed to points forming a Horton set in $H$.
	\item The points on any vertical line spanned by points of $G$ are perturbed to points forming a rotated copy of a Horton set in $H$.
\end{enumerate}
\end{defini}

As already mentioned in the introduction, squared Horton sets have been defined by Valtr~\cite{Va92}. 
A way to construct them is also presented in~\cite{BV2004}. 
For self-containment, we describe a construction similar to the one in~\cite{BV2004} here, for $n$ being a squared integer:

Let $H_x$ be a Horton set of $\sqrt{n}$ points such that the $x$-coordinates are the integers $1,\dots, \sqrt{n}$, and
its $y$-coordinates are in $[-\varepsilon_x,+\varepsilon_x]$ for some arbitrarily small $0 < \varepsilon_x <1/4$. 
This can be accomplished by a suitable linear transformation of a Horton set with points evenly spaced in the $x$-coordinate. Let $H_y$ be a Horton set defined as before, for some $0 <\varepsilon_y < \varepsilon_x$  and rotated $90$ degrees, 
so that the $y$-coordinates of $H_y$ are the integers $1,\dots,\sqrt{n}$ and its $x$-coordinates
are in $[-\varepsilon_y,+\varepsilon_y]$. 

Further, let 
$H:=\{(x_1+x_2,y_1+y_2): (x_1,y_1) \in H_x \textrm{ and } (x_2,y_2) \in H_y \}$
be the Minkowski sum of $H_x$ and $H_y$ and let $G:=\{(i,j): 1 \leq i,j \leq \sqrt{n}\}.$ 
Note that for every point $p$ of $H$ there is a unique point $(i,j) \in G$  at distance at most $\varepsilon:=\varepsilon_x+\varepsilon_y$ of $p$.
Thus, $H$ is an $\varepsilon$-perturbation of $G$. Let $\pert: G \to H$ be the map that sends each such $(i,j) \in G$ to its unique closest $p\in H$.
Observation~\ref{obs:hortonsubsets} implies that if $\varepsilon_x$ is chosen small enough and $\varepsilon_y$ is sufficiently
smaller than $\varepsilon_x$, then $H$ is a squared Horton set since the following conditions hold.
\begin{enumerate}
 \item For every triple $p_i,p_j,p_k$ of non-collinear points of $G$, the orientation
 of $(p_i,p_j,p_k)$ and $({\pert}(p_i),{\pert}(p_j),{\pert}(p_k))$
 is the same.
 \item For every non-vertical straight line $\ell$ that is spanned by points of $G$, ${\pert}(\ell \cap G)$ is a Horton set.
 \item For every vertical straight line $\ell$ that is spanned by points of $G$, ${\pert}(\ell \cap G)$ is a rotated copy of a Horton set.
\end{enumerate}

When reasoning about a squared Horton set $H$ 
we repeatedly reason about structures in $H$ and the according structures in the underlying unperturbed grid $G$ in parallel. 
To relate structures in $G$ with their perturbed structures in $H$, we will denote by $\pert$ the map that is induced by the $\varepsilon$-perturbation that transforms $G$ to $H$.

The following lemma is a direct consequence of the definition of squared Horton sets.
\begin{restatable}{lemma}{sqHortonUnion}\label{lemma:sqHortonUnion}
	Let $H=\pert(G)$ be a squared Horton set 
	and let $\ell$ and $\ell'$ be two parallel lines spanned by $G$. Then $\pert((G\cap \ell) \cup (G \cap \ell'))$ is a (rotated copy of a) Horton set.
\end{restatable}
\begin{proof}
Assume without loss of generality that $\ell$ and $\ell'$ are not vertical and that $\ell$ is below $\ell'$ 
	(otherwise, rotate $H$ accordingly).
Let $H_0=\pert(G\cap \ell)$ and $H_1=\pert(G\cap \ell')$. 
Since $H$ is a squared Horton set, $H_0$ and $H_1$ are both Horton sets.
Furthermore, $H_0$ is deep below $H_1$ since every line passing through two points of $H_0$ is below any point of $H_1$.
Similarly $H_1$ is high above $H_0$.
Therefore, $H_0 \cup H_1$ is a rotated copy of a Horton set.
\end{proof}

A triangle in a squared Horton set $H=\pert(G)$ either corresponds to a triangle in $G$ or to a set of three collinear points in $G$.
In the following, we denote the latter as a \emph{degenerate} triangle. 
Further, for any empty triangle $\pert(\Delta)$ in $H$, $\Delta$ is either degenerate or interior-empty in $G$,
due to the fact that $\pert$ is the map of an $\varepsilon$-perturbation.

Let $\Delta$ be a (possibly degenerate) triangle with vertices in $G$. 
Let $e$ be an edge of $\Delta$ and let $p$ be the vertex of $\Delta$
opposite to $e$. 
We say that the \emph{height of $\Delta$ w.r.t.~$e$} is zero
if $p$ is on the straight line spanned by $e$; otherwise, it is one plus
the number of lines between $e$ and $p$, that are  parallel to $e$, and  that contain
points of the integer grid $\mathbb{Z} \times \mathbb{Z}$. We call the area bounded by two such neighboring lines a \emph{strip}.
The \emph{height} of $\Delta$ is the minimum of the heights w.r.t.\ its edges
and the edge defining the height of $\Delta$ is the \emph{base edge}. 

We review a few basics results regarding line and line segments with points in the integer grid.
\begin{lemma}\label{lem:integer_lines}
 Let $\ell$ be a line containing at least two points $a,b \in \mathbb{Z} \times \mathbb{Z}$.
 Then there exists $d>0$, such that any two consecutive points along  $\ell$ in $\mathbb{Z} \times \mathbb{Z}$
 are at a distance $d$ of each other.
\end{lemma}
\begin{proof}
 Suppose that $\ell$ is not vertical, as otherwise the result holds with $d=1$.
 Since $a,b$ are points of $\ell$, the slope $m$ of $\ell$ is a rational number. 
  Let $\ell'$ be the translation
 of $\ell$ by the vector $-a$, so that the point $a$ in $\ell$ is translated to the origin.
 Let $r/s:=m$, with $r,s$ relative prime integers. Note that $\ell'$ has equation
 $y=\frac{r}{s} x$. Thus,   for $(x,y) \in \ell'$, we have that $(x,y) \in \mathbb{Z} \times \mathbb{Z}$ if
 and only if $x$ is an integer multiple of $s$. In this case $y$ is an integer multiple of $r$. Thus, the distance between
 any two consecutive points along $\ell'$ with integer coordinates is equal to
 \[d=\sqrt{s^2+r^2}.\]
 Since $\ell'$ is a translation of $\ell$ by a vector
 in $\mathbb{Z} \times \mathbb{Z}$ , every pair of consecutive points, along $\ell$, of $\ell \cap  \mathbb{Z} \times \mathbb{Z}$ 
 are at a distance $d$ of each other. 
\end{proof}

\begin{corollary}\label{cor:integer_segment}
Let $a,b \in \mathbb{Z} \times \mathbb{Z}$. Let $\ell$ be a line
 parallel to $ab$ and containing a point of  $\mathbb{Z} \times \mathbb{Z}$.
 Then every line segment, $e$, contained in $\ell$, of
 length at least $|ab|$ contains at least one point of $\mathbb{Z} \times \mathbb{Z}$.
 Moreover, if $e$ has an endpoint in $\mathbb{Z} \times \mathbb{Z}$, the $e$
 contains at least two points of $\mathbb{Z} \times \mathbb{Z}$.
\end{corollary}
\begin{proof}
 Let $a'$ be a point $\ell \cap  \mathbb{Z} \times \mathbb{Z}$. 
 Let $a'b'$ be a line segment parallel to $ab$. Note that $b' \in  \mathbb{Z} \times \mathbb{Z}$.
 Let $d$ be as in Lemma~\ref{lem:integer_lines} for $\ell$.
 Note that $e \ge |ab| =|a'b'| \ge d$. Since every pair of consecutive points, along $\ell$, of $\ell \cap  \mathbb{Z} \times \mathbb{Z}$ 
 are at a distance $d$ of each other, the result follows.
\end{proof}

\begin{lemma}\label{lem:cutting_lines}
 Let $m$ be a rational number. Let $L$ be the set of lines with slope $m$ that pass through some point
 of $\mathbb{Z} \times \mathbb{Z}$. Let $B$ be the intersection points of the lines in $L$ and
 the $x$-axis. Then there exists $d >0$, such that every two points in $B$, that are consecutive along
 the $x$-axis, are at a distance $d$ of each other.
\end{lemma}
\begin{proof}
 The lines in $L$ have equation of the form $y=mx+b$.
 Therefore, 
 \[B = \{b: b=y-mx \textrm{ for some } (x,y) \in \mathbb{Z} \times \mathbb{Z}\}.\]
  Thus, $B$ is the image of the group homomorphism from $\mathbb{Z} \times \mathbb{Z}$ to $\mathbb{Q}$
 that maps $(x,y)$ to $y-mx$. Since  $\mathbb{Z} \times \mathbb{Z}$ is finitely generated, $B$ is also finitely generated.
 As every finitely generated subgroup of $\mathbb{Q}$ is cyclic, there exists a rational number $d$ such that
 $B=\{nd: n \in \mathbb{Z}\}.$
 The result follows.
\end{proof}

\begin{restatable}{lemma}{triangleHeight}\label{lemma:triangle_height}
Any interior-empty triangle of $G$ has height at most $2$.
\end{restatable}
\begin{proof}
Let $\Delta$ be a triangle with vertices in $G$. We first show that

\LabelQuote{if two edges of $\Delta$ have each at least two interior points in $G$, then
there is a point of $G$ in the interior of $\Delta$.}{$\ast$}

\begin{figure}
\centering
\includegraphics[page=4]{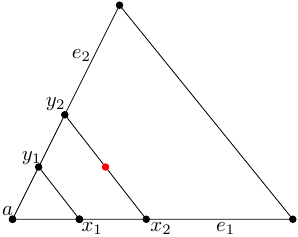}
\caption{A triangle $\Delta$ with two edges containing two interior points of $G$ each and the induced point inside $\Delta$.}
\label{fig:height2:case3}
\end{figure}

Let $e_1$ and $e_2$ be two edges of $\Delta$, each with at least two interior points
in~$G$. Let $a$ be the vertex of $\Delta$ common to $e_1$ and $e_2$. Let $x_1$ and $x_2$ be
the points closest and second closest to $a$ in $e_1 \cap \mathbb{Z} \times \mathbb{Z}$, respectively.
Let $y_1$ and $y_2$ be
the points closest and second closest to $a$ in $e_2 \cap \mathbb{Z} \times \mathbb{Z}$, respectively. See in Figure~\ref{fig:height2:case3}.
By Lemma~\ref{lem:integer_lines}, there exist $d_1,d_2 >0$ such that 
$|ax_1|=|x_1x_2|=d_1$ and $|ay_1|=|y_1y_2|=d_2$. This implies
that $x_2y_2$ is parallel to and twice the length of $x_1y_1$. Therefore,
the midpoint of $x_2y_2$ is in $G$, which proves ($\ast$).

Now assume for the contrary that $\Delta$ is interior-empty and has height at least 3.
Let $a,b$ and $c$, be the vertices of $\Delta$, and let $e:=ab$.
Since $\Delta$ has height at least 3,
 there exist at least two lines parallel to $e$ each containing a point
of $\mathbb{Z} \times \mathbb{Z}$ and crossing through the interior of $\Delta$.
 Of these lines let $\ell_1$ and $\ell_2$ be the lines
closest and second closest to $e$, respectively. Let $\ell'$ be the line parallel to $bc$ and containing $a$. 
Let $x$ be the point of intersection of $\ell'$ and $\ell_1$;
let $y$ be the point of intersection of $ac$ and $\ell_1$; and
let $y'$ be the point of intersection between $bc$ an $\ell_1$. See Figure~\ref{fig:height2:case2}.
Since $ab$ is parallel to $xy'$ and they have the same length, by Corollary~\ref{cor:integer_segment} there exists a point $p \in \mathbb{Z} \times \mathbb{Z}$ on $xy'$.
We may assume that $yy'$ does not contain
points of $G$ in its interior as otherwise $\Delta$ is not interior-empty.
Therefore, $p$ is either on the line segment $xy$ or $p=y'$.

\begin{figure}
\centering
\includegraphics[page=3]{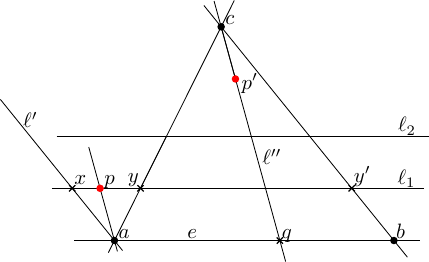}
\caption{If a point $p \in G$ is in the interior of $xy$ then the triangle $abc$ is not empty. (Points of $G$ are marked with disks, points which are not neccessarily in $G$ are marked with crosses.)}
\label{fig:height2:case2}
\end{figure}

Suppose next that $p$ is in the interior of $xy$. Let $\ell''$ be the line
parallel to $pa$ and containing $c$. Let $q$ be the point of intersection
between $\ell''$ and $ab$, as depicted in Figure~\ref{fig:height2:case2}. Note that $|pa|<|cq|$. By Corollary~\ref{cor:integer_segment},
$cq$ contains a point of $G$ in its interior, and $\Delta$ is not interior-empty.

If $p=x$, then $y'$ has integer coordinates and hence is a point of $G$. 
So it remains to consider $p \in \{y, y'\}$.
Suppose that $p=y$ or $y'$. Let $e'$ be the side of $\Delta$ that contains $p$,
and let $q$ be the endpoint of $e'$ distinct from $c$.
Let $p'$ be the intersection point of $e'$ and $\ell_2$. 
By Lemma~\ref{lem:cutting_lines}, we have
that $|qp|=|pp'|$. Since $p$ has integer coordinates, then so does $p'$, and $e'$ contains
two interior points in $G$.
We repeat the previous arguments now with $e=e'$ to conclude that either $\Delta$ is interior non-empty
or $\Delta$ has an edge distinct from $e'$ with two interior points in $G$. In the latter case, we are done
by ($\ast$).
\end{proof}

For the proof of our next statement, we use the Euler's totient function $\phi$.
For a given integer $d$, $\phi(d)$ is the number of integers at most $d$ that are relative primes with $d$. 
Clearly,  $\phi(d) \leq d$. A segment with endpoints in the integer grid is \emph{primitive} if it does not contain
any integer grid point in its integior.
It is well-known that  for $d>1$:
\begin{itemize}
\item  $2 \cdot \phi(d)$ is the number of points $(d,a)$ with $|a|<|d|$ on the integer grid such that the segment from the origin to the point $(d,a)$ is primitive; and 
\item $2 \cdot \phi(d)$ is the number of points $(a,d)$ with $|a|<|d|$ on the integer grid such that the segment from the origin to the point $(a,d)$ is primitive.
\end{itemize}
We use the following lemma to get asymptotic bounds later on.
\begin{lemma}\label{lemma:calc_horton_stabs}
Let $n\geq 1$ be the square of an integer. Then
 \begin{eqnarray*}
 \sum_{d=1}^{ \sqrt{n} } \phi(d) ( \sqrt{n} /d) \cdot \log_2( \sqrt{n} /d) & = & O(n).
 \end{eqnarray*}
\end{lemma}

\begin{proof}
We use the  bounds $\phi(d) \le d$ and $n!\geq (n/e)^n$. The latter follows by  Stirling's formula. 
 \begin{eqnarray*}
 \sum_{d=1}^{ \sqrt{n} } \phi(d) \cdot ( \sqrt{n} /d ) \cdot \log_2( \sqrt{n}/d )
 & \le  & \sum_{d=1}^{ \sqrt{n} } \sqrt{n} \cdot \log_2( \sqrt{n}/d ) \\
 & = & \sqrt{n} \cdot \log_2 \left( \prod_{d=1}^{\sqrt{n}} \sqrt{n}/d  \right) \\
 & = & \sqrt{n} \cdot \log_2  \left( (\sqrt{n})^{\sqrt{n}} / \sqrt{n}! \right) \\
	 & \le &  \sqrt{n} \cdot \log_2  \left( \sqrt{n}^{\sqrt{n}} / (\sqrt{n}/e)^{\sqrt{n}} \right) \\
	 & = & \sqrt{n} \cdot\log_2  \left( e^{ \sqrt{n} } \right) \\
	 & \le & (\sqrt{n})^2 \cdot \log_2(e)  \\
 & = & O(n). 
 \end{eqnarray*}
\end{proof}

\begin{restatable}{theorem}{sqHortonStabs}\label{thm:horton_stabs}
	Let $H=\pert(G)$ be a squared Horton set of $n$ points. Then every point of the plane stabs $O(n)$ empty triangles of $H$. 
\end{restatable}
\begin{proof}
  Obviously, no point of $H$ can stab any empty triangle of $H$.
 Consider an arbitrary point $q \in \mathbb{R}^2\setminus H$. 
	Every empty triangle $\pert(\Delta)$ in $H$ corresponds to an interior-empty triangle $\Delta$ in $G$. By Lemma~\ref{lemma:triangle_height}, $\Delta$ has height at most 2.
 We separately count the empty triangles of different heights in $H$ that possibly contain $q$. 
	We consider each such triangle $\pert(\Delta)$ by the slope of the base edge of $\Delta$ in $G$.
 
We start with the triangles of height zero; these triangles correspond to degenerate triangles in $G$.
Let $\Delta$ be a degenerate triangle  in $G$, with slope $m$, such that $\pert(\Delta)$ is stabbed by $q$. Let $\ell$ be  the line
that contains $\Delta$. Let $\ell'$ be a line distinct from $\ell$ and parallel to $\ell$.
Let $\Delta'$ be a degenerate triangle in $G$, contained in $\ell'$. By Property $1$ of the definition
of Squared Horton sets, the convex hulls of $\pert(\Delta)$ and $\pert(\Delta')$ do not intersect.
In particular $\pert(\Delta')$ is not stabbed by $q$.
This implies that for every possible slope, $m$, spanned by points in $G$, there exists at most one line 
containing all degenerate triangles $\Delta$ of $G$, with slope $m$, such that $\pert(\Delta)$ is stabbed by $q$. 

 Suppose that $|m| < 1$. Let $d >0$ be the distance, in $x$-direction, between two consecutive points of $\ell \cap G$.
Note that  $d$ is an integer satisfying $1 \le d \le \sqrt{n}$;  thus,  $|\ell \cap G| \le  \sqrt{n}/d+1$.
Let $(d,a)$ be the vector defined by two consecutive integer grid points in $\ell$. Note that
since $|m| < 1$, we have that $|a| < |d|$. Therefore, $m$ has at most $2 \cdot \phi(d)$ different possible values. 
By a similar argument, if $|m| >1$ we have that:  $|\ell \cap G| \le  \sqrt{n}/d+1$ for 
 some integer  $1 \le d \le \sqrt{n}$; and  $m$ has  at most $2 \cdot \phi(d)$ different possible values.
 Therefore, $m$  has at most $4\cdot\phi(d)$ different possible values. By Properties $2$ and $3$ of the definition
 of Squared Horton sets, $\pert(\ell \cap G)$ forms a Horton set. 
  Hence, by Lemma~\ref{lemma:hortoninterior}, the number of empty triangles in $\pert(\ell \cap G)$ that contain $q$ is bounded by $O(\sqrt{n}/d \cdot \log(\sqrt{n}/d))$. 
  Summing this bound over all possible slopes and applying Lemma~\ref{lemma:calc_horton_stabs}, we obtain an upper bound of 
	\[ \sum_{d=1}^{\sqrt{n}} 4 \cdot \phi(d) \cdot O\big(\sqrt{n}/d \cdot \log(\sqrt{n}/d)\big) 
	= O\Big(\sum_{d=1}^{\sqrt{n}} \phi(d) \cdot \sqrt{n}/d \cdot \log(\sqrt{n}/d)\Big) = O(n)\]
for the number of empty triangles of height zero that can be stabbed by $q$.

 Suppose that $m \in (0^\circ, 45^\circ] \cup  (-90^\circ, -45^\circ]$. Let $d >0$ be the distance, in $x$-direction, between two consecutive points of $\ell \cap G$.
Note that  $d$ is an integer satisfying $1 \le d \le \sqrt{n}$;  thus,  $|\ell \cap G| \le  \sqrt{n}/d+1$.
 Further note that $m$ has at most $2 \cdot \phi(d)$ different possible values. Suppose that 
 $m \in  (-90^\circ, -45^\circ] \cup (45^\circ, 90^\circ]$. By a similar argument (with the roles
 of the $x$ and $y$ directions interchanged) we have that:  $|\ell \cap G| \le  \sqrt{n}/d+1$ for 
 some integer  $1 \le d \le \sqrt{n}$; and  $m$ has  at most $2 \cdot \phi(d)$ different possible values.
 Therefore, $m$  has at most $4\cdot\phi(d)$ different possible values. By Properties $2$ and $3$ of the definition
 of Squared Horton sets, $\pert(\ell \cap G)$ forms a Horton set. 
  Hence, by Lemma~\ref{lemma:hortoninterior}, the number of empty triangles in $\pert(\ell \cap G)$ that contain $q$ is bounded by $O(\sqrt{n}/d \cdot \log(\sqrt{n}/d))$. 
  Summing this bound over all possible slopes and applying Lemma~\ref{lemma:calc_horton_stabs}, we obtain an upper bound of 
	\[ \sum_{d=1}^{\sqrt{n}} 4 \cdot \phi(d) \cdot O\big(\sqrt{n}/d \cdot \log(\sqrt{n}/d)\big) 
	= O\Big(\sum_{d=1}^{\sqrt{n}} \phi(d) \cdot \sqrt{n}/d \cdot \log(\sqrt{n}/d)\Big) = O(n)\]
for the number of empty triangles of height zero that can be stabbed by $q$.

We next bound the number of triangles of height one that contain $q$.
Consider two empty triangles $\Delta$ and $\Delta'$ of height one in $G$ whose base edges are parallel and 
for which $q$ lies in both $\pert(\Delta)$ and $\pert(\Delta')$.
Let $S$ and $S'$ be the points of $G$ that are on the boundary of the parallel strips containing $\Delta$ and $\Delta'$, respectively.
Since the strips are parallel and $q \in \Conv(\pert(S)) \cap \Conv(\pert(S'))$, it follows that $S \cap S' \neq \emptyset$.
In other words, $\Delta$ and $\Delta'$ lie either in the same strip or in two neighboring strips of $G$.
So for each possible slope of the base edge of $\Delta$, there are at most two strips in $G$ which can contain $\Delta$ such that $q$ stabs $\pert(\Delta)$.

Similar as before, for each $1 \le d < \sqrt{n}$ there are at most $4 \cdot \phi(d)$ possible slopes 
for the base edge of $\Delta$ and the according strip $S$ containing $\Delta$ has at most $2(\sqrt{n}/d+1)$ points of $G$.
By Lemma~\ref{lemma:sqHortonUnion}, $\pert(S)$ forms a Horton set in $H$. Hence, by Lemma~\ref{lemma:hortoninterior}, 
the number of empty triangles in $\pert(S)$ that contains $q$ is upper bounded by $O(\sqrt{n}/d \cdot \log(\sqrt{n}/d))$.
Summing up over all possible slopes and the according strips, we obtain an upper bound of 
	\[ \sum_{d=1}^{\sqrt{n}} 4 \cdot \phi(d) \cdot 2 \cdot O\big(\sqrt{n}/d \cdot \log(\sqrt{n}/d)\big) 
	= O\Big(\sum_{d=1}^{\sqrt{n}} \phi(d) \cdot \sqrt{n}/d \cdot \log(\sqrt{n}/d)\Big) = O(n)\]
for the number of empty triangles of height one that can be stabbed by $q$.

Finally, we consider the triangles of height $2$. 
Let $\ell$ be the supporting line of the base edge of a triangle $\Delta$ in $G$ of height $2$ 
such that $\pert(\Delta)$ is empty and contains $q$. 
Let $S$ be the set of points of $G$ in the double strip bounded by $\ell$ that contains $\Delta$ 
and let $p$ be the corner of $\Delta$ that does not lie on its base edge. 
Note that all interior-empty triangles with height 2, base edge on $\ell$, and third corner $p$ are pairwise interior-disjoint. Hence $\Delta$ is the only such triangle for which $\pert(\Delta)$ is empty and contains $q$. 

Now consider a triangle $\Delta'$ of height $2$ for which $\pert(\Delta')$ is empty and contains $q$ and whose base edge is parallel to the one of $\Delta$. 
Let $S'$ be the set of points of $G$ in the double strip parallel to $\ell$ that contains $\Delta'$.
Since the double-strips of $\Delta$ and $\Delta'$ are parallel and $q \in \Conv(\pert(S)) \cap \Conv(\pert(S'))$, it follows that $S \cap S' \neq \emptyset$.
In other words, the two double strips must be identical, overlapping, or neighboring
So for each possible slope of the base edge of $\Delta$, there are at most three double-strips in $G$ which can contain $\Delta$ such that $q$ stabs $\pert(\Delta)$. Hence at most five lines of that slope could contain the third vertex $p$ of $\Delta$. 

For each $1 \le d \le \sqrt{n}$ there are at most $4 \cdot \phi(d)$ possible slopes for the base edge of $\Delta$. 
Each of the at most five lines of such a slope has at most $\sqrt{n}/d+1$ points of $G$, each of which could be $p$.
Summing up over all possible slopes and the according lines and slopes and using $\phi(d)\leq d$, we obtain an upper bound of
	\[ \sum_{d=1}^{\sqrt{n}} 4 \cdot \phi(d) \cdot 5(\sqrt{n}/d + 1) 
	= O\Big(\sqrt{n} \sum_{d=1}^{\sqrt{n}} \phi(d)/d \Big) = O(n)\]
for the number of empty triangles of height $2$ that can be stabbed by $q$.

Adding up the bounds for the three different triangle heights yields an upper bound of $3\cdot O(n) = O(n)$ on the total number of empty triangles in $H$ that contain $q$, which completes the proof. 
\end{proof}
	
The following result on the number of empty triangles incident to a fixed point of a squared Horton set
is proven in a similar way as Theorem~\ref{thm:horton_stabs}.

\begin{restatable}{lemma}{sqHortonIncident}\label{lemma:sqHorton_incident}
Every point of a squared Horton set of $n$ points is incident to $O(n)$ empty triangles.
\end{restatable}

\begin{proof}
 Let $H$ be a squared Horton set of $n$ points. We will show that
 every point $p$ of $H$ is incident to $O(n)$ empty triangles of $H$.

 We start with the number of such triangles of height zero.
	For each such triangle $\Delta$, $\pert^{-1}(p)$ has to be on the same
	line as $\pert^{-1}(\Delta)$. For each $1 \le d \le \sqrt{n}$, there are at most $4 \cdot \phi(d)$ lines spanned by $G$
	through $\pert^{-1}(p)$, each with at most $\sqrt{n}/d+1$ points of $G$. 
	For each such line $\ell$, the point set $\pert(\ell \cap G)$ is a Horton set
	which by Lemma~\ref{lemma:hortonincident} has at most $O((\sqrt{n}/d) \cdot \log(\sqrt{n}/d))$ empty triangles incident to $p$.
	Hence, summing over all possible slopes, the number of empty triangles in $H$ incident to $p$ and having height zero is at most
 \begin{eqnarray*}
	 \sum_{d=1}^{\sqrt{n}} \phi(d) \cdot O\big((\sqrt{n}/d) \cdot \log_2(\sqrt{n}/d)\big) & = & O(n). \\
 \end{eqnarray*}
 
We next consider the triangles of height one. 
	For $p$ to be incident to $\Delta$, $\pert^{-1}(p)$ has to lie on the boundary of the strip defining the height of $\pert^{-1}(\Delta)$.
For each slope there are at most two relevant such strips for $p$, each spanning a Horton set in $H$. 
	By a similar counting as above we again get a bound of $O(n)$ on the number of empty triangles in $H$ of height one incident to $p$.

Finally, we consider the triangles of height $2$.
For each slope there are at most $3$ grid lines in $G$ that could contain
	the basis of an interior-empty triangle incident to $\pert^{-1}(p)$, 
namely, the line $\ell$ containing $\pert^{-1}(p)$ and the lines $\ell_1,\ell_2$ which are $2$ apart from $\ell$.
	Further, for each such double-strip, the number of empty triangles in $H$ that is incident to $p$ is bounded from above by the number of points on the boundary of the double-strip.
	Hence summing up over all slopes and relevant double-strips gives $O(n)$ empty triangles in $H$ that have height $2$ and are incident to $p$, which completes the proof.
\end{proof}

\section{{\wasylozenge}-squared Horton sets}\label{sec:diamondssquaredhortonset}

\begin{figure}
	\centering
	\includegraphics[scale=1]{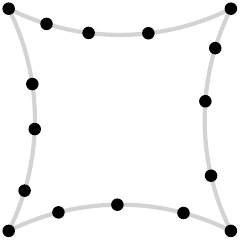}
	\caption{ A {\wasylozenge} point set.	}
    \label{fig:wasylozenge}
\end{figure}
We denote by {\wasylozenge} a point set obtained by placing four points on the corners of a square and adding further 
points along four slightly concave arcs between adjacent corners, 
such that on each arc there is almost the same number of points.
An example is depicted in Figure~\ref{fig:wasylozenge}.

\begin{defini} 
Let $H$ be a squared Horton set with $m$ points.
Let $H_{\wasylozenge}$ be the set we obtain by replacing every point of $H$ 
 by a small $\wasylozenge$ with $k$ points.
We denote the points of $H$ by $p_i$ and the corresponding $\wasylozenge$ by $\wasylozenge_i$.
Then $H_{\wasylozenge}$ is a \text{{\wasylozenge}-\emph{squared}} \emph{Horton set} if the following properties hold.
\begin{enumerate}
\item\label{enum:orient} For any pairwise different $i, j, l \in \{1,\dots, m\}$, any point triple $q_i \in \wasylozenge_i$, $q_j\in \wasylozenge_j$, $q_l \in \wasylozenge_l$ has the same orientation as 
$p_i,p_j,p_l$.
\item\label{enum:arc} The arcs of each~$\wasylozenge_i$ are such that for any $\wasylozenge_j$ with $i\neq j$ there is an arc of $\wasylozenge_i$ and arc of $\wasylozenge_j$ which form a convex set.
$\Conv(\wasylozenge_i \cup \wasylozenge_j)\cap \Conv(\wasylozenge_i \cup \wasylozenge_l)$.
\item\label{point_on_line} For any five pairwise different $a, b, c, d, e \in \{1, \dots, m\}$ the following holds
$\Conv(\wasylozenge_a\cup\wasylozenge_b)\cap\Conv(\wasylozenge_a\cup\wasylozenge_c)\cap\Conv(\wasylozenge_d\cup\wasylozenge_e)=\emptyset$. 
\item\label{three_cross} For any six pairwise different $a, b, c, d, e, f \in \{1, \dots, m\}$ the following holds
$\Conv(\wasylozenge_a\cup\wasylozenge_b)\cap\Conv(\wasylozenge_c\cup\wasylozenge_d)\cap\Conv(\wasylozenge_e\cup\wasylozenge_f)=\emptyset$. 
\end{enumerate}
\end{defini}
Observe, that $H_{\wasylozenge}$ has $n=km$ points if $H$ has $m$ points and $\wasylozenge$ consists of $k$ points.
Since the points of $H$ are in general position it is possible to choose the $\wasylozenge$ small enough, such that the Properties~\ref{enum:orient},~\ref{point_on_line} and \ref{three_cross} hold.
Property~\ref{enum:arc} holds if all $\wasylozenge$ are aligned in the same way.

\begin{restatable}{lemma}{wasHortonEmpty}\label{lemma:nbr_empty}
Let $H_{\wasylozenge}$ be a ${\wasylozenge}$-squared Horton set, where the underlying squared Horton set has $m$ points and each $\wasylozenge$ consists of $k$ points.
Then the number of empty triangles in $H_{\wasylozenge}$ is $\Theta(m^2k^3)$.
\end{restatable}

\begin{proof}
We split the empty triangles of $H_{\wasylozenge}$ into three groups,
depending on the number of different $\wasylozenge$ subsets of $H_{\wasylozenge}$ that contain vertices of a triangle. 

\emph{Case 1.} Triangles spanned by three points of $\wasylozenge_i$, for $i \in \{1,\dots,m\}$.
Each $\wasylozenge_i$ spans $O(k^3)$ such empty triangles. Summing up over the $m$ different subsets 
$\wasylozenge_1,\dots,\wasylozenge_m$ yields $O(m k^3)$ empty triangles of $H_{\wasylozenge}$ for this case.

\emph{Case 2.} Triangles spanned by two points in $\wasylozenge_i$ and one point in $\wasylozenge_j$, for $i\neq j \in \{1,\dots,m\}$.
Note that every such triangle does not have any point of $\wasylozenge_l$ with 
$l \neq i,j$ in its interior due to the first property of $\wasylozenge$-squared Horton sets.
	There are $\Theta(m^2)$ pairs $(\wasylozenge_i, \wasylozenge_j)$.
For each of $\wasylozenge_i$ and $\wasylozenge_j$, there are at most $k$ choices for a vertex of an empty triangle.
This means, we have $O(m^2k^3)$ empty triangles in this case.
On the other hand, due to the third property of $\wasylozenge$-squared Horton sets, an arc of $\wasylozenge_i$ and an arc of $\wasylozenge_j$ form a convex point set. This convex point set is empty by construction. Further, each of the arcs has at least $k/4$ points. 
This gives us $\Omega(m^2k^3)$ empty triangles in this case. 
So $H_{\wasylozenge}$ contains $\Theta(m^2k^3)$ empty triangles which are spanned by two $\wasylozenge$s.

\emph{Case 3.} Triangles spanned by one point in each of $\wasylozenge_i, \wasylozenge_j, \wasylozenge_l$, for pairwise different $i, j, l \in \{1,\dots,m\}$.
Then $p_i,p_j,p_l$ is an empty triangle of $H$. 
For each of $p_i$, $p_j$, and $p_l$, we have 
at most $k$ choices for a point of the corresponding $\wasylozenge$ such that the resulting triangle of $H_{\wasylozenge}$ 
is empty. As $H$ has $\Theta(m^2)$ empty triangles, we have $O(m^2k^3)$ empty triangles of $H_{\wasylozenge}$ for this case.
\end{proof}

\begin{restatable}{lemma}{lemWasHortonStabbed}\label{lemma:nbr_stabbed}
Let $H_{\wasylozenge}$ be a ${\wasylozenge}$-squared Horton set, where the underlying squared Horton set has $m$ points and each $\wasylozenge$ consists of $k$ points.
Then every point of the plane stabs $O(mk^3)$ empty triangles of $H_{\wasylozenge}$.
\end{restatable}

\begin{proof}
We fix a stabbing point $s$.
We split the empty triangles into three groups: three points of one $\wasylozenge$, two points in one $\wasylozenge$ and the third one in a different $\wasylozenge$, and all three points in different $\wasylozenge$s.

\emph{Case 1.} All points in one $\wasylozenge$.
If all points are in one $\wasylozenge$, then a point $s$ stabs at most $k^3$ such empty triangles since the convex hulls of the $\wasylozenge$s do not intersect.

\emph{Case 2.} Two points in one $\wasylozenge$ and the third point $q$ in a different $\wasylozenge$.
We draw a half ray $\ell$ starting from $q$ and through $s$.
If $a,b$ are points of a $\wasylozenge$, which is not intersected by $\ell$ then $s$ does not stab the triangle $abq$.
Further, $\ell$ intersects at most one other $\wasylozenge$ since the points of the underlying squared Horton set are in general position.
So we have $mk$ choices for $q$ and $O(k^2)$ choices for the other two points of the triangle, so that the triangle is stabbed by $s$. 

\emph{Case 3.} All points in different $\wasylozenge$s.
We merge the points of each $\wasylozenge_i$ into one point $p_i$.
The result of this merging is a squared Horton set $H$.
Further we define a point $s'$ with the following properties:
\begin{itemize}
\item $s':=s$ if $s \notin \Conv(\wasylozenge_i \cup \wasylozenge_j)$ for any $i,j\in \{1,\dots,m\}$,

\item $s':=p_i$, with $i \in \{1,\dots,m\}$ if there exist two different $j, l \in \{1,\dots,m\}\backslash\{i\}$ with ${s \in \Conv(\wasylozenge_i \cup \wasylozenge_j)\cap \Conv(\wasylozenge_i \cup \wasylozenge_l)}$,

\item $s':=p_ip_j\cap p_ap_b$, with pairwise distinct $a,b,i,j \in \{1,\dots,m\}$ if it holds that $s \in \Conv(\wasylozenge_i \cup \wasylozenge_j)\cap \Conv(\wasylozenge_a \cup \wasylozenge_b)$.

\item $s'\in p_ip_j$ if $s \in \Conv(\wasylozenge_i \cup \wasylozenge_j)$ such that $s'$ is on the same side of the line $p_ap_b$ as $s$ for any $a,b \in \{1,\dots,m\}\backslash\{i,j\}$.
\end{itemize}

If $s'=p_i$ then by Property~\ref{point_on_line} there do not exist $a,b \in \{1, \dots, m\}\backslash\{i\}$ such that $s\in\Conv(\wasylozenge_a \cup \wasylozenge_b)$.
Further, if $s'=p_ip_j\cap p_ap_b$ then by Property~\ref{three_cross} there do not exist $c, d\in \{1, \dots, m\}\backslash\{a,b,i,j\}$ such that $s\in\Conv(\wasylozenge_c \cup \wasylozenge_d)$.
So $s'$ is well defined.

Observe, that $s'$ is on the same side of the line $p_ap_b$ as $s$ or on the line $p_ap_b$ for any $a,b \in \{1,\dots,m\}$.
This means, if $s$ stabs $q_iq_jq_l$ with $q_i \in \wasylozenge_i$, $q_j \in \wasylozenge_j$ and $q_l \in \wasylozenge_l$, then $s'$ is in the interior or on the boundary of the corresponding triangle $p_ip_jp_l$. 
We have three cases:

\emph{Case 3a.} $s'$ is neither a point of $H$ nor a on a line segment spanned by two points of $H$.
By Theorem~\ref{thm:horton_stabs} $s'$ stabs $O(m)$ empty triangles of $H$.
Let $p_i,p_j,p_k$ be the points of a triangle of $H$ stabbed by $s'$.
There are at most $k$ choices to select a point of $\wasylozenge_i,\wasylozenge_j,\wasylozenge_k$, respectively.
So $s$ stabs $O(mk^3)$ such empty triangles.

\emph{Case 3b.} $s'$ is a point of $H$.
By Lemma~\ref{lemma:sqHorton_incident}, $s'$ is incident to $O(m)$ empty triangles of $H$.
Let $p_i,p_j,p_k$ be the points of a triangle of $H$ incident to~$s$.
There are at most $k$ choices to select a point of $\wasylozenge_i,\wasylozenge_j,\wasylozenge_k$, respectively.
So $s$ stabs $O(mk^3)$ such empty triangles.

\emph{Case 3c.} $s'$ is on a line $p_ip_j$.
We define two points $s'_1$ and $s'_2$ such that both are close to $s'$ and $s'_1$ and $s'_2$ are on different sides of any line $p_ap_b$ passing through $s'$, for $a,b \in \{1,\dots,m\}$.
Observe that any triangle containing $s'$ (in the interior or on the boundary) also contains $s'_1$ or $s'_2$.
Also observe that $s'_1$ and $s'_2$ are neither points of $H$ nor on a line spanned by points of $H$.
So $s'_1$ and $s'_2$ stab $O(m)$ empty triangles, respectively.
There are at most $k$ choices to select a point of $\wasylozenge_i,\wasylozenge_j,\wasylozenge_k$, respectively.
Again we get, that $s$ stabs $O(mk^3)$ such empty triangles.

Adding up all cases $s$ stabs $O(mk^3)$ triangles.
\end{proof}

With these lemmata we can finally show our main result.

\begin{namedproof}{Proof of Theorem~\ref{thm:LensSquaredHortonSets}}
Consider a $\wasylozenge$-squared Horton set where the underlying squared Horton set consists of $m=n^{\alpha}$ points and each of the $\wasylozenge$'s consist of $k=n^{1-\alpha}$ points.
By Lemma~\ref{lemma:nbr_empty} there are $\Theta(m^2k^3)=\Theta(n^{3-\alpha})$ empty triangles in $S$.
By Lemma~\ref{lemma:nbr_stabbed} every point stabs $O(mk^3)$ empty triangles.
Hence every point stabs $O(n^{3-2\alpha})$ empty triangles.
\end{namedproof}

\paragraph{Acknowledgements.} 
Research on this work has been initiated at a workshop of the H2020-MSCA-RISE project 73499 - CONNECT, held in Barcelona in June 2017.
We thank all participants for the good atmosphere as well as for discussions on the topic.
\bibliographystyle{abbrv} 
\bibliography{selectionbib}

\begin{thebibliography}{10}

\bibitem{alon1992point}
N.~Alon, I.~B{\'a}r{\'a}ny, Z.~F{\"u}redi, and D.~J. Kleitman.
\newblock Point selections and weak $\varepsilon$-nets for convex hulls.
\newblock {\em Combinatorics, Probability and Computing}, 1(3):189--200, 1992.

\bibitem{SecondSelectionAronov}
B.~Aronov, B.~Chazelle, H.~Edelsbrunner, L.~J. Guibas, M.~Sharir, and
  R.~Wenger.
\newblock Points and triangles in the plane and halving planes in space.
\newblock {\em Discrete and Computational Geometry}, 6(5):435--442, 1991.

\bibitem{BARANY-82}
I.~B{\'a}r{\'a}ny.
\newblock A generalization of {C}arath{\'e}odory's theorem.
\newblock {\em Discrete Mathematics}, 40(2):141 -- 152, 1982.

\bibitem{barany1990number}
I.~B{\'a}r{\'a}ny, Z.~F{\"u}redi, and L.~Lov{\'a}sz.
\newblock On the number of halving planes.
\newblock {\em Combinatorica}, 10(2):175--183, 1990.

\bibitem{high_degree_barany}
I.~B\'{a}r\'{a}ny and G.~K\'{a}rolyi.
\newblock Problems and results around the {E}rd{\H{o}}s-{S}zekeres convex
  polygon theorem.
\newblock In {\em Discrete and computational geometry ({T}okyo, 2000)}, volume
  2098 of {\em Lecture Notes in Computer Science}, pages 91--105. Springer,
  Berlin, 2001.

\bibitem{BaranyMR13}
I.~B{\'{a}}r{\'{a}}ny, J.~Marckert, and M.~Reitzner.
\newblock Many empty triangles have a common edge.
\newblock {\em Discrete and Computational Geometry}, 50(1):244--252, 2013.

\bibitem{BV2004}
I.~B\'ar\'any and P.~Valtr.
\newblock Planar point sets with a small number of empty convex polygons.
\newblock {\em Studia Scientiarum Mathematicarum Hungarica}, 41(2):243--266,
  2004.

\bibitem{boros1984number}
E.~Boros and Z.~F{\"u}redi.
\newblock The number of triangles covering the center of an n-set.
\newblock {\em Geometriae Dedicata}, 17(1):69--77, 1984.

\bibitem{EppsteinImprovedBound}
D.~Eppstein.
\newblock Improved bounds for intersecting triangles and halving planes.
\newblock {\em Journal of Combinatorial Theory. Series A}, 62(1):176--182,
  1993.

\bibitem{high_degree_erdos}
P.~Erd\H{o}s.
\newblock On some unsolved problems in elementary geometry.
\newblock {\em Matematikai Lapok. New Series}, 2(2):1--10, 1992.

\bibitem{horton_1983}
J.~D. Horton.
\newblock Sets with no empty convex 7-gons.
\newblock {\em Canadian Mathematical Bulletin}, 26(4):482–484, 1983.

\bibitem{kat}
M.~Katchalski and A.~Meir.
\newblock On empty triangles determined by points in the plane.
\newblock {\em Acta Mathematica Hungarica}, 51(3-4):323--328, 1988.

\bibitem{MatousekBook-2002}
J.~Matou{\v s}ek.
\newblock {\em Lectures on discrete geometry}, volume 212 of {\em Graduate
  Texts in Mathematics}.
\newblock Springer-Verlag, New York, 2002.

\bibitem{EppsteinRevisited}
G.~Nivasch and M.~Sharir.
\newblock Eppstein's bound on intersecting triangles revisited.
\newblock {\em Journal of Combinatorial Theory. Series A}, 116(2):494--497,
  2009.

\bibitem{Va92}
P.~Valtr.
\newblock Convex independent sets and 7-holes in restricted planar point sets.
\newblock {\em Discrete and Computational Geometry}, 7:135--152, 1992.

\bibitem{Va95}
P.~Valtr.
\newblock On the minimum number of empty polygons in planar point sets.
\newblock {\em Studia Scientiarum Mathematicarum Hungarica}, 30:155--163, 1995.

\end{thebibliography}

\end{document}